\RequirePackage{fix-cm}
\documentclass[smallextended]{svjour3}       
\smartqed  
\usepackage{graphicx}
\usepackage{amsmath,amssymb}

\usepackage{tikz}
\usetikzlibrary{calc,positioning,arrows}
\usepackage{float}

\usepackage[linesnumbered,boxed,ruled]{algorithm2e} 
\SetAlgoInsideSkip{medskip}

\newcommand{\N}{\mathtt{N}}

\def\join{\small\textcircled{\tiny 1}}
\def\co-join{\small\textcircled{\tiny 0}}

\begin{document}

\title{On the complete width and edge clique cover problems\thanks{Parts of this paper appeared in proceedings of the 21st Annual International Computing and Combinatorics Conference (COCOON 2015), August 4-6, 2015, Beijing, China~\cite{cocoon2015}.}
}


\author{Van Bang Le \and
        Sheng-Lung Peng
}


\institute{V.B. Le \at
              Universit\"at Rostock, Institut f\"ur Informatik, Germany\\
              \email{van-bang.le@uni-rostock.de} \\
           \and
           S.-L. Peng \at
              Department of Computer Science and Information Engineering,\\
              National Dong Hwa University, Hualien 974, Taiwan\\
              \email{slpeng@mail.ndhu.edu.tw}
}

\date{Received: date / Accepted: date}

\maketitle

\begin{abstract}
A complete graph is the graph in which every two vertices are adjacent. For a graph $G=(V,E)$, the complete width of $G$ is the minimum $k$ such that there exist $k$ independent sets $\N_i\subseteq V$, $1\le i\le k$, such that the graph $G'$ obtained from $G$ by adding some new edges between certain vertices inside the sets $\N_i$, $1\le i\le k$, is a complete graph. The complete width problem is to decide whether the complete width of a given graph is at most $k$ or not.
In this paper we study the complete width problem. We show that the complete width problem is NP-complete on $3K_2$-free bipartite graphs and polynomially solvable on $2K_2$-free bipartite graphs and on $(2K_2,C_4)$-free graphs. As a by-product, we obtain the following new results: the edge clique cover problem is NP-complete on $\overline{3K_2}$-free co-bipartite graphs and polynomially solvable on $C_4$-free co-bipartite graphs and on $(2K_2, C_4)$-free graphs. We also give a characterization for $k$-probe complete graphs which implies that the complete width problem admits a kernel of at most $2^k$ vertices. This provides another proof for the known fact that the edge clique cover problem admits a kernel of at most $2^k$ vertices. Finally we determine all graphs of small complete width $k\le 3$.

\end{abstract}

\section{Introduction}
Let $G=(V,E)$ be a simple and undirected graph. A subset $U\subseteq V$ is an {\em independent set}, respectively, a {\em clique} if no two, respectively, every two vertices of $U$ are adjacent. The complete graph with $n$ vertices is denoted by $K_n$. The path and cycle with $n$ vertices of length $n-1$, respectively, of length $n$, is denoted by $P_n$, respectively, $C_n$. For a vertex $v\in V$ we write $N(v)$ for the set of its neighbors in $G$. A \emph{universal} vertex $v$ is one such that $N(v)\cup\{v\}=V$.
For a subset $U\subseteq V$ we write $G[U]$ for the subgraph of $G$ induced by $U$ and $G-U$ for the graph $G[V \setminus U]$; for a vertex $v$ we write $G-v$ rather than $G- \{v\}$.

Given a graph class $\mathcal{C}$, a graph $G=(V,E)$ is called a \emph{probe $\mathcal{C}$ graph} (or \emph{$\mathcal{C}$ probe graphs\/}) if there exists an independent set $\N\subseteq V$ (\emph{i.e.}, a set of \emph{nonprobes}) and a set of new edges $E'\subseteq\binom{\N}{2}$ between certain nonprobe vertices such that the graph $G'=(V,E\cup E')$ is in
the class $\mathcal{C}$, where $\binom{\N}{2}$ stands for the set of all 2-element subsets of $\N$.
A graph $G=(V,E)$ with a \emph{given} independent set $\N\subseteq V$ is said to be a \emph{partitioned probe $\mathcal{C}$ graph} if there exists a set $E'\subseteq\binom{\N}{2}$ such that the graph $G'=(V,E\cup E')$ is in the class $\mathcal{C}$. In both cases, $G'$ is called a $\mathcal{C}$ {\em embedding} of $G$. Thus, a graph is a (partitioned) probe $\mathcal{C}$ graph if and only if it admits a $\mathcal{C}$ embedding.
The most popular case is the class $\mathcal{C}$ of interval graphs. The study of probe interval graphs was motivated from certain problems in physical mapping of DNA in the computational biology; see, {\em e.g.},~\cite{ChaChaKloPen,GolTre}.

Recently, the concept of probe graphs has been generalized as a width parameter of graph class in~\cite{ChaHunKloPen}.
Let $\mathcal{C}$ be a class of graphs.
The {\em $\mathcal{C}$-width\/} of a graph $G$
is the minimum number $k$ of independent sets
$\N_1,\dots, \N_k$ in $G$ such that there exists
an embedding $G' \in \mathcal{C}$ of $G$ such that for
every edge $xy$ in $G'$ which is not an edge of $G$ there
exists an $i$ with $x,y \in \N_i$.
A collection of such $k$ independent sets $\N_i, i = 1, \ldots, k$, is called a {\em $\mathcal{C}$ witness} for $G$.
In the case $k=1$, $G$ is a \emph{probe $\mathcal{C}$-graph}.
The {\em $\mathcal{C}$-width\/} problem asks for a given graph $G$ and an integer $k$ if the $\mathcal{C}$-width of $G$ is at most $k$. Graphs of $\mathcal{C}$-width $k$ are also called $k$-probe $\mathcal{C}$-graph. Note that graphs in $\mathcal{C}$ are, by convenience, $0$-probe $\mathcal{C}$-graphs.

In \cite{ChaHunKloPen}, the complete width and block-graph width have been investigated. The authors proved that, for fixed $k$, graphs of complete width $k$ can be characterized by finitely many forbidden induced graphs. Their proof is however not constructive. They also showed, implicitly, that complete width $k$ graphs and block-graph width $k$ graphs can be recognized in cubic time. The case $k=1$, {\em e.g.}, probe complete graphs and probe block graphs, has been discussed in depth in~\cite{LePeng12}. The case $k=2$ is discussed in~\cite{LePeng14}.

Graphs that do not contain an induced subgraph isomorphic to a graph $H$ are called \emph{$H$-free\/}.  More generally, a graph is \emph{$(H_1,\ldots, H_t)$-free\/} if it does not contain an induced subgraph isomorphic to one of the graphs $H_1,\ldots, H_t$.
For two graphs $G$ and $H$, we write $G+H$ for the disjoint union of $G$ and $H$, and for an integer $t\ge 2$, $tG$ stands for the disjoint union of $t$ copies of $G$. The complete $k$-partite with $n_i$ vertices in color class $i$ is denoted by $K_{n_1,\ldots, n_k}$.
For graph classes not defined here see, for example,~\cite{BraLeSpi,ChaChaKloPen,Golumbic}.

In this paper we study the \textsc{complete width} problem (given $G$ and $k$, is the complete width of $G$ at most $k$?). We show that
\begin{itemize}
\item \textsc{complete width} is NP-complete, even on $3K_2$-free bipartite graphs, and
\item computing the complete width of a $2K_2$-free bipartite graph (chain graph), and more generally, of a $(2K_2,K_3)$-free graph can be done in polynomial time,
\item computing the complete width of a $2K_2$-free chordal graph (split graph), and more generally, of a $(2K_2,C_4)$-free graph can be done in polynomial time,
\item \textsc{complete width} admits a kernel with at most $2^k$ vertices. That is, any instance $(G,k)$ of \textsc{complete width} can be reduced in polynomial time to an equivalent instance $(G',k')$ of \textsc{complete width} with $k'\le k$ and $G'$ has at most $2^k$ vertices. In particular, \textsc{complete width} is fixed-parameter tractable with respect to parameter $k$.
\end{itemize}
Moreover, we give structural characterizations for graphs of complete width at most $3$.

In the next section we point out a relation between complete width and the more popular notion of edge clique cover of graphs. Then we prove our results in the last four sections. As we will see, it follows from our results on complete width that edge clique cover is NP-complete on $\overline{3K_2}$-free co-bipartite graphs and is polynomially solvable on $C_4$-free co-bipartite graphs.

\section{Complete width and edge clique cover}\label{sec:ecc}
An \emph{edge clique cover} of a graph $G$ is a family of cliques (complete subgraphs) such
that each edge of $G$ is in at least one member of the family.
The minimal cardinality of an edge clique cover is the \emph{edge clique cover number}, denoted by $\theta_\mathrm{e}(G)$.

The \textsc{edge clique cover} problem, the problem of deciding if $\theta_\mathrm{e}(G)\le k$, for a given graph $G$ and an integer $k$, is NP-complete \cite{Holyer,KSW,Orlin}, even when restricted to
graphs with maximum degree at most six~\cite{Hoover}, or planar graphs \cite{CM}.
\textsc{edge clique cover} is polynomially solvable for graphs with maximum degree at most five~\cite{Hoover}, for line graphs~\cite{Orlin,Pullman}, for chordal graphs~\cite{MWW,Ray}, and for circular-arc graphs~\cite{HsuTsa}.

In~\cite{KSW} it is shown that approximating the edge clique covering number within a constant factor smaller than two is NP-hard. In \cite{GGHN}, it is shown that \textsc{edge clique cover} is fixed-parameter tractable with respect to parameter $k$; see also \cite{CKPPW,CPP} for more recent discussions on the parameterized complexity aspects.

We write $cow(G)$ to denote the complete width of the graph $G$. As usual, $\overline{G}$ denotes the complement of $G$. In \cite{ChaHunKloPen}, the authors showed that \textsc{complete width} is NP-complete on general graphs, by observing that

\begin{proposition}[\cite{ChaHunKloPen}]\label{prop:cow-ecc}
For any graph $G$,
$cow(G)=\theta_\mathrm{e}(\overline{G})$
\end{proposition}

Proposition~\ref{prop:cow-ecc} and the known results about \textsc{edge clique cover} imply:
\begin{theorem}\label{thm:ecc}
\begin{itemize}
\item[\em (1)] Computing the complete width is NP-hard, and remains NP-hard when restricted to graphs of minimum degree at least $n-7$, and to co-planar graphs.
\item[\em (2)] Computing the complete width of graphs of minimum degree at least $n-6$ and of co-chordal graphs can be done in polynomial time.
\end{itemize}
\end{theorem}

In \cite{CHL}, it is conjectured that \textsc{edge clique cover}, and thus \textsc{complete width}, is NP-complete for $P_4$-free graphs (also called cographs).

\medskip
We close this section by the following basic facts about complete width, which will be useful later.

\begin{proposition}\label{prop:basic}
Let $G$ be a graph.
\begin{itemize}
\item[\em (1)] If $v$ is a universal vertex in $G$, then $cow(G)=cow(G-v)$.
\item[\em (2)] Let $G$ have no universal vertices. Suppose $u$ and $v$ are two vertices in $G$ with $N(u)=N(v)$. Then\\
$cow(G)=\begin{cases}
cow(G-u)+1, & \text{if $v$ is universal in $G-u$}\\
cow(G-u),   & \text{otherwise}
\end{cases}
$
\end{itemize}
\end{proposition}
\begin{proof} (1): This is obvious.

\noindent
(2): Assume first that $v$ is universal in $G-u$. Then, clearly, the independent set $\{u,v\}$ belongs to any complete witness for $G$. Since $G$ has no universal vertices, $G-u-v$ is not a complete graph, {\em i.e.}, $cow(G-u-v)\geq 1$. Hence $cow(G)-1=cow(G-u-v)=cow(G-u)$, where the second equality follows from (1).

Suppose now that $v$ is not universal in $G-u$, and let $\N_1, \ldots, \N_k$ be a complete witness for $G-u$ with $k=cow(G-u)$. Set $\N_i'=\N_i$ if $v\not\in \N_i$ and $\N_i'=\N_i\cup\{u\}$ if $v\in \N_i$. Clearly, $\N_1',\ldots,\N_k'$ are independent sets in $G$. Furthermore, $\N_1',\ldots,\N_k'$ form a complete witness for $G$: Consider two non-adjacent vertices $x\not= y$ in $G$. If $u\not\in \{x,y\}$, then $x$ and $y$ belong to some $\N_i$, hence to some $\N_i'$. So, let $u=x$, say. If $v\not= y$, then $v$ and $y$ are non-adjacent in $G-u$ (as $N(u)=N(v)$), hence $v$ and $y$ belong to some $\N_i$. Hence $u=x$ and $y$ belong to $\N_i'=\N_i\cup\{u\}$. If $v=y$, then, as $v$ is not universal in $G-u$, $v$ is non-adjacent to some $z\in G-u-v$. Hence $v=y$ and $z$ belong to some $\N_i$, and so, $u=x$ and $y$ belong to $\N_i'=\N_i\cup\{u\}$.

Thus, $\N_1',\ldots,\N_k'$ form a complete witness for $G$, as claimed. Therefore, $cow(G)$ $\le$ $cow(G-u)$, hence $cow(G)=cow(G-u)$.
\qed
\end{proof}

Thus, by Proposition~\ref{prop:basic}, we often assume that, when discussing complete width, all graphs in question have no universal vertices and $N(u)\not= N(v)$ for any non-adjacent vertices $u, v$.

\section{Computing complete width is hard for $3K_2$-free bipartite graphs}
A bipartite graph $G=(V,E)$ is a graph whose vertex set $V$ can be partitioned into two sets $X$ and $Y$ such that for any edge $xy\in E$, $x\in X$ and $y\in Y$. Bipartite graphs without induced cycles of length at least six are called \emph{chordal bipartite\/}.
A \emph{biclique cover} of a graph $G$ is a family of complete bipartite subgraphs of $G$ whose edges cover the edges of $G$. The \emph{biclique cover number}, also called the \emph{bipartite dimension}, of $G$ is the minimum number of bicliques needed to cover all edges of $G$.

Given a graph $G$ and a positive integer $k$, the \textsc{biclique cover} problem of $G$ asks whether the edges of $G$ can be covered by at most $k$ bicliques.
The following theorem is well known.

\begin{theorem}[\cite{Muller,Orlin}]
\textsc{biclique cover} is NP-complete on bipartite graphs, and remains NP-complete on chordal bipartite graphs.
\end{theorem}

For convenience, a bipartite graph $G=(V,E)$ with a bipartition $V=X\cup Y$ into independent sets $X$ and $Y$ is denoted as $G=(X+Y, E)$. Let $BC(G)=(X+Y,F)$, where $F=\{xy\mid x\in X,\,y\in Y, \mbox{ and } xy\not\in E\}$. We call $BC(G)$ the \emph{bipartite complement} of $G=(X+Y, E)$. Note that $BC(C_6)=3K_2$ and $BC(C_8)=C_8$. Hence if $G$ is chordal bipartite, then $BC(G)$ is $(3K_2, C_8)$-free bipartite.

In \cite{ChaHunKloPen}, the authors showed that the complete width problem is NP-complete on general graphs. We now establish our main theorem for sharpening that result of \cite{ChaHunKloPen}.

\begin{theorem}\label{thm:cow}
\textsc{complete width} is NP-complete on bipartite graphs, and remains NP-complete on $(3K_2, C_8)$-free bipartite graphs.
\end{theorem}
\begin{proof}
We prove this theorem by reducing \textsc{biclique cover} to \textsc{complete width}.

Let $(G, k)$ be an input instance of the biclique cover problem, where $G=(X+Y,E)$ is a bipartite graph. We construct an input instance $(G', k')$ of the complete width problem as follows.

\begin{itemize}
\item $G'$ is the bipartite graph obtained from the bipartite complement $BC(G)=(X+Y,F)$ of $G$ by adding two new vertices $x$ and $y$ and adding all edges between $x$ and vertices in $Y\cup\{y\}$ and between $y$ and vertices in $X\cup\{x\}$.

More formally, $G'=(X'+Y',F')$ with $X'=X\cup \{x\}, Y'=Y\cup \{y\}$, and $F'=F\cup \{xu\mid u\in Y\cup \{y\}\}\cup \{yv\mid v\in X\cup \{x\}\}$.

\item Set $k':=k+2$.
\end{itemize}

We claim that the biclique cover number of $G$ is at most $k$ if and only if the complete width of $G'$ is at most $k'=k+2$.

First, let $\{B_i \mid 1\le i\le k\}$ be a biclique cover of $G$, where $B_i=(X_i+Y_i,E_i)$ with $X_i\subseteq X, Y_i\subseteq Y$. Then, as each $B_i$ is a biclique in $G$, each $\N_i=X_i\cup Y_i$ is an independent set in $G'$. Set $\N_{k+1}:= X'$ and $\N_{k+2}:= Y'$. Then it is easy to check that the $k'=k+2$ independent sets $\N_i$, $1\le i\le k+2$, form a complete witness for $G$. That is, $cow(G')\le k'$.

Conversely, let $\{\N_i \mid 1\le i\le k+2\}$ be a complete witness for $G'$. Then we may assume that
$$ x, y\not\in \N_i, 1\le i\le k.
$$
(To see this, consider a vertex $u\in X$. As $\{\N_i \mid 1\le i\le k+2\}$ is a complete witness for $G'$, $u$ and $x$ must belong to $\N_t$ for some $t\in\{1,\ldots, k+2\}$.
Therefore, $\N_t\subseteq X\cup\{x\}=X'$ because $x$ is adjacent to all vertices in $Y'$.
Clearly, we can replace $\N_t$ by $X'$ and, if $x\in \N_i$ for some $i\not= t$, replace $\N_i$ by $\N_i\setminus\{x\}$ to obtain a new witness such that $\N_t=X'$ and $x$ is contained only in $\N_t$.
Similarly, there is some $s$ such that $\N_s=Y'$ and $y$ is contained only in $\N_s$.
By re-numbering if necessary, we may assume that $t=k+1$ and $s=k+2$.)

Thus, by construction of $G'$, $\N_1, \ldots, \N_k$ are independent sets in $BC(G)$ and form a complete witness for $BC(G)$. Therefore, $B_i=G[\N_i]$, $1\le i\le k$, are bicliques in $G$ forming a biclique cover of $G$. That is, $\theta_\mathrm{e}(G)\le k$.

Note that if $G$ is chordal bipartite, then the bipartite graph $G'$ cannot contain $3K_2$ and $C_8$ as induced subgraphs. \qed
\end{proof}

Theorem~\ref{thm:cow} and Proposition~\ref{prop:cow-ecc} imply the following new NP-completeness result for \textsc{edge clique cover}.
\begin{corollary}\label{cor:ecc}
\textsc{edge clique cover} is NP-complete on $(\overline{3K_2},\overline{C_8})$-free co-bipar\-ti\-te graphs.
\end{corollary}

\section{Polynomially solvable cases}
In this section we establish some cases in which \textsc{complete width} can be solved in polynomial time. Actually, in each of these cases we will show that the complete width of the graphs under consideration can be computed in polynomial time.

\subsection{$2K_2$-free bipartite graphs}
Bipartite graphs without induced $2K_2$ are known in literature under the name {\em chain graphs} (\cite{Yannakakis}) or {\em difference graphs} (\cite{HamPelSun}). They can be characterized as follows.
\begin{proposition}[see \cite{MahPel}]\label{prop-chain}
A bipartite graph $G=(X+Y, E)$ is a chain graph if and only if for all vertices $u, v\in X$, $N(u)\subseteq N(v)$ or $N(v)\subseteq N(u)$.
\end{proposition}
\begin{theorem}\label{thm:chain}
The complete width of a chain graph can be computed in polynomial time.
\end{theorem}
\begin{proof}
Let $G=(X+Y, E)$ be a $2K_2$-free bipartite graph with at least three vertices. By Proposition~\ref{prop:basic}, we may assume that for any pair of vertices $u, v$ of $G$, $N(u)\not=N(v)$. Thus, $|X|\ge 2, |Y|\ge 2$, and $G$ has at most one non-trivial connected component and at most one trivial component which is then the unique isolated vertex of $G$. Let us also assume that the isolated vertex (if any) of $G$ belongs to $X$. By Proposition~\ref{prop-chain}, the vertices of $X$ can be numbered $v_1, v_2, \ldots, v_{|X|}$ such that $N(v_1)\subset N(v_2)\subset \cdots \subset N(v_{|X|})=Y$. Thus, $G$ is disconnected if and only if $v_1$ is the isolated vertex of $G$ if and only if $N(v_1)=\emptyset$. Clearly, such a numbering can be computed in polynomial time.

Write $\N_i =\{v_1,\ldots, v_i\}\cup (Y\setminus N(v_i))$, $1\le i\le |X|$. Since $N(v_j)\subset N(v_i)$ for $j<i$, $\N_i$ is an independent set, and since $N(v_{|X|})=Y$, $\N_{|X|}=X$. In case $N(v_1)\not=\emptyset$, let $\N_{|X|+1}=Y$. Note that in the case that $N(v_1)=\emptyset$, \emph{i.e.}, $v_1$ is the isolated vertex of $G$, $\N_1=Y\cup\{v_1\}$.

We claim that
$$
cow(G)=\begin{cases}
      |X|, & \text{ if } N(v_1)=\emptyset\\
      |X|+1, & \text{ otherwise}
\end{cases}
$$

Moreover, $\N_1,\ldots, \N_{|X|}$ and $\N_{|X|+1}$ (if $N(v_1)\not=\emptyset$) together form a complete witness for $G$.

\noindent\medskip
\emph{Proof of the Claim:\/} First, to see that the collection of the independent sets $\N_1$, \ldots, $\N_{|X|}$ and $\N_{|X|+1}$ (if $N(v_1)\not=\emptyset$) is a complete witness for $G$, let $u, v$ be two non-adjacent vertices of $G$. If $u, v \in X$, say $u=v_i$ and $v=v_j$ for some $1\le i <j\le |X|$, then $u, v\in \N_j$.
If $u\in X$ and $v\in Y$, say $u=v_i$ for some $1\le i\le |X|$, then $u,v\in \N_i$.
So let $u, v\in Y$. In this case, let $i\le j$ be the smallest integers such that $u\in N(v_i), v\in N(v_j)$. If $i>1$ then $u, v\not\in N(v_1)$, hence $u,v\in\N_1$. Thus, let $u\in N(v_1)$. Then, in particular $N(v_1)\not=\emptyset$ and hence $u,v\in \N_{|X|+1}=Y$.

In particular, $cow(G)$ is at most the right hand side stated in the claim.

Next, observe that the claim is clearly true in case $|X|=2$. So, let $|X|\ge 3$. Note that in $G-v_1$, $N(v_2)$ is not empty, hence by induction, $cow(G-v_1)=|X\setminus\{v_1\}|+1 = |X|$ and $\N_1'= \N_2\setminus\{v_1\},\ldots, \N_{|X|-1}'=\N_{|X|}\setminus\{v_1\}$ and $\N_{|X|}'=\N_{|X|+1}=Y$ form a complete witness for $G-v_1$. Now, if $N(v_1)=\emptyset$ then $cow(G)\ge cow(G-v_1)=|X|$, hence $cow(G)=|X|$. So, let $N(v_1)\not=\emptyset$. In this case, for any $u\in N(v_2)\setminus N(v_1)$ and any maximal independent set $I$ of $G$ containing $v_1$ and $u$, $\N_i'\not\subseteq I$. Thus, $cow(G)\ge cow(G-v_1)+1=|X|+1$, hence $cow(G)=|X|+1$.

The proof of the claim is completed, hence Theorem~\ref{thm:chain}.
\qed
\end{proof}

Theorem~\ref{thm:chain} and Proposition~\ref{prop:cow-ecc} imply the following corollary.
\begin{corollary}
The edge clique cover number of a $C_4$-free co-bipartite graph can be computed in polynomial time.
\end{corollary}

\subsection{$(2K_2, K_3)$-free graphs}
We extend Theorem~\ref{thm:chain} on $K_2$-free bipartite graphs by showing that \textsc{complete width} is polynomially solvable for large class of $2K_2$-free triangle-free graphs.

\begin{theorem}\label{thm:2K2K3-free}
The complete width of a $(2K_2,K_3)$-free graph can be computed in polynomial time.
\end{theorem}
\begin{proof}
Let $G$ be a $(2K_2, K_3)$-free graph. If $G$ has no induced $C_5$, then $G$ is $2K_2$-free bipartite, hence we are done by Theorem~\ref{thm:chain}.

So let $G$ contain an induced $C_5$, say $C=v_1v_2v_3v_4v_5v_1$. By Proposition~\ref{prop:basic}, we may assume that $N(u)\not=N(v)$ for any non-adjacent vertices $u$ and $v$ of $G$. We will see that $C$ is a connected component of $G$. Let $H$ be the connected component of $G$ containing $C$. If $H\not=C$, then there is some vertex $v\in H-C$ adjacent to some vertex in $C$, say $v_1$. Since $G$ is $(2K_2,K_3)$-free, $v$ is non-adjacent to $v_2,v_5$ and adjacent to $v_3$ or $v_4$ but not both. Let $v$ be adjacent to $v_3$, say. Now, as $N(v)\not=N(v_2)$, there is a vertex $u$ adjacent to $v$ and non-adjacent to $v_2$, say. But then $G[C+u+v]$ has a $K_3$ or a $2K_2$, a contradiction. Thus $H=C$ and as $G$ is $2K_2$-free, $C$ is the only non-trivial connected component of $G$, hence $cow(G)=5$.
\qed
\end{proof}

\subsection{Split graphs}

A \emph{split graph} is one whose vertex set can be partitioned into a clique $Q$ and an independent set $S$. For convenience, a split graph is denoted as $G=(Q+S,E)$. It is well known that split graphs can be characterized as follows.

\begin{proposition}[\cite{FolHam}]\label{prop:split}
The following statements are equivalent for any graph $G$.
\begin{itemize}
\item[\em (i)] $G$ a split graph;
\item[\em (ii)] $G$ is a $(2K_2, C_4, C_5)$-free graph;
\item[\em (iii)] $G$ is a $2K_2$-free chordal graph;
\item[\em (iv)] $G$ and $\overline{G}$ are chordal.
\end{itemize}
\end{proposition}

In particular, split graphs are complements of chordal graphs. Hence, by Theorem~\ref{thm:ecc} (2), computing the complete width of split graphs can be done in polynomial time. Below, however, we give a simple and direct way for doing this. Moreover, our solution will be useful for computing the complete width of pseudo split graphs. The class of pseudo split graphs are not necessarily co-chordal and properly contains all split graphs.

In the following, by Proposition~\ref{prop:basic}, we may consider the split graphs $G=(Q+S,E)$ with no universal vertex.

\begin{theorem}\label{split-size}
For a split graph $G=(Q+S,E)$ with no universal vertex, the complete width of $G$ is either $|Q|$ or $|Q|+1$.
\end{theorem}
\begin{proof}
Assume that the complete width of $G$ is $k$. That is, there is an embedding $G'$ of $G$ such that for every edge $xy$ in $G'$ but not in $G$ there are independent sets $\N_1,\dots,\N_k$ in $G$ such that $\{x,y\}\subseteq \N_i$ for some $i$. By the definition, $G[Q]$ is a clique. Thus it is impossible that there are two vertices of $Q$ in the same $\N_i$ for $1\le i\le k$. That is, each $\N_i$ contains at most one vertex in $Q$. Therefore, the complete width of $G$ is at least $|Q|$.

On the other hand, for each vertex $v\in Q$, let $\N_v=V(G)\setminus N(v)$. Then, each $\N_v$, $v\in Q$, is an independent set. Further, for each $\N_v$, we can fill edges $vu$, $u\in \N_v-v$. Finally, for the final set $S$, we make $G[S]$ a clique by filling edge $xy$ for any two vertices $x,y\in S$. The resulting graph is a complete graph. That is, the complete width of $G$ is at most $|Q|+1$. This completes the proof.\qed
\end{proof}

By Theorem \ref{split-size}, there are only two cases for determining the complete width of a split graph. For the split graph $G=(Q+S,E)$, let $\N_v=V(G)\setminus N(v)$ for $v\in Q$. We have the following lemma.

\begin{lemma}\label{split-cond}
For a split graph $G=(Q+S,E)$ with no universal vertex, if for any two vertices $x,y\in S$, there is an $\N_v$, $v\in Q$, such that $x,y\in \N_v$, then the complete width of $G$ is $|Q|;$ otherwise it is $|Q|+1$.
\end{lemma}
\begin{proof}
Assume that for any two vertices $x,y\in S$, there is an $\N_v$, $v\in Q$ such that $x,y\in \N_v$. We show that the complete width of $G$ is $|Q|$. Without loss of generality, we assume all the $\N_v$'s are ordered as the sequence of $\N_1,\N_2,\dots,\N_{|Q|}$. For completing $G$ into $K_n$, for each $\N_v$, we fill the edges $vu$, $u\in (\N_v\cap S)$. Furthermore, assume that $\N_i$ is the last set that contains $x$ and $y$ for any two vertices $x,y\in S$. That is, $\{x,y\}\subseteq \N_i$ but $\{x,y\}\not\subseteq \N_j$ for each $j>i$. Then the edge $xy$ is filled in $\N_i$. By assumption, every edge in $\overline{G}[S]$ can be filled in some $\N_i$. Thus the complete width of $G$ is $|Q|$.

On the other hand, if no $\N_i$ contains $x$ and $y$ for some $x,y\in S$, then there is no way to fill $x,y$ in $\N_1,\N_2,\dots, \N_{|Q|}$. Therefore the complete width of $G$ is $|Q|+1$.\qed
\end{proof}

By Lemma \ref{split-cond}, for any two vertices $x,y\in S$, we can check whether there is a vertex $v\in Q$ such that both $xv$ and $yv$ are in $E$ or not. By using adjacency matrix of $G$, all the work can be done in $O(n^3)$ time. Thus, we have the following theorem.

\begin{theorem}\label{thm:cow-split}
The complete width of a split graph can be computed in polynomial time.
\end{theorem}

\subsection{Pseudo-split graphs}
Graphs without induced $2K_2$ and $C_4$ are called {\em pseudo-split graphs}. By Proposition~\ref{prop:split}, the class of pseudo-split graphs properly contains the class of split graphs. Note that a pseudo-split graph may contain an induced $C_5$, hence it might not be co-chordal. Pseudo-split graphs can be characterized as follows.
\begin{theorem}[\cite{BHPT,MafPre}]\label{thm:pseudo-split}
A graph is pseudo-split if and only if its vertex set can be partitioned into
three sets $Q, S, C$ such that $Q$ is a clique, $S$ is an independent set, $C$ induces a $C_5$ or is empty, $xy$ is an edge for each $x\in Q$ and each $y\in C$, and there are no edges between $S$ and $C$.
\end{theorem}
Note that it can be recognized in linear time if a graph is a pseudo split graph, and if so, a partition stated in Theorem~\ref{thm:pseudo-split} can be found in linear time~\cite{MafPre}.
\begin{theorem}\label{thm:cow-pseudo-split}
The complete width of a pseudo-split graph can be computed in polynomial time.
\end{theorem}
\begin{proof}
Let $G=(V,E)$ be a pseudo-split graph without universal vertices. Let $V=Q+S+C$ be a partition as in Theorem~\ref{thm:pseudo-split}. We may assume that $C\not=\emptyset$ otherwise we are done by Theorem~\ref{thm:cow-split}.

So let $C$ be the induced $C_5=v_1v_2v_3v_4v_5v_1$. Then, clearly, the $|Q|+5$ independent sets $V-N(v)$, $v\in Q$, and $S\cup\{v_i,v_{i+2}\}$ (indices are taken modulo $5$), $1\le i\le 5$, can be used for completing $G$. Thus, by Theorem~\ref{split-size}, and by noting that $cow(C_5)=5$, we have $cow(G)=|Q|+5$.\qed
\end{proof}

Theorem~\ref{thm:cow-pseudo-split} and Proposition~\ref{prop:cow-ecc} imply the following corollary (note that the complement of a pseudo-split is also a pseudo-split graph).
\begin{corollary}
The edge clique cover number of a pseudo-split graph can be computed in polynomial time.
\end{corollary}

\section{Problem kernel}\label{section:kernel}

Parameterized complexity deals with NP-hard problems whose instances come
equipped with an additional integer parameter $k$. The objective is to design algorithms
whose running time is $f(k)\cdot \text{poly}(n)$ for some computable
function $f$ depending only on $k$ and some polynomial $\text{poly}(\cdot)$. Problems admitting such algorithms are called \emph{fixed-parameter tractable}. See, {\em e.g.},~\cite{CFKLMPPS} for more information.
It is well known that fixed-parameter tractable problems are exactly those problems having a kernel. Here, a {\em kernel} is an algorithm that, given an instance $(x,k)$ of the problem with a fixed parameter $k$, outputs in polynomial time in $|x|+k$ an `equivalent' instance $(x',k')$ of the same problem such that $|x'|, k'\le g(k)$ for some computable function $g$ depending only on $k$.

As mentioned, \textsc{complete width} and \textsc{edge clique cover} are NP-comp\-le\-te in general and fixed parameter tractable with respect to $k$. In \cite{ChaHunKloPen}, an fpt-algorithm for \textsc{complete width} was given, based on the monadic second order logic. In \cite{GGHN}, it was shown that \textsc{edge clique cover} admits a problem kernel of at most $g(k)=2^k$ vertices.

In this section, we give a characterization of $k$-probe complete graphs, which will imply that \textsc{complete width} admits a problem kernel of at most $2^k$ vertices. With Proposition~\ref{prop:cow-ecc}, this provides an alternative way to see that \textsc{edge clique cover} admits a problem kernel of at most $2^k$ vertices (\cite{GGHN}).

To this end, we first construct, for a given integer $k$, a prototype for graphs with complete width $k$.  Write $[k]=\{1,\ldots, k\}$ and let $\mathcal{P}[k]$ be the set of all subsets of $[k]$. We define the graph $G[k]$ as follows: $V(G[k])=\mathcal{P}[k]$, $E(G[k])=\{\{M,L\}\mid M\cap L=\emptyset\}$. Thus, the vertices of $G[k]$ are the subsets of $\{1,\ldots,k\}$ and two subsets are adjacent in $G[k]$ whenever they are disjoint. Let $G\star H$ be the {\em join} of $G$ and $H$ obtained from $G+H$ by adding all possible edges $xy$ between any vertex $x \in G$ and any vertex $y \in H$. Then, $G[1]$ is the clique $K_2$, $G[2]=(K_2+K_1)\star K_1$, $G[3]=(\text{Net}+K_1)\star K_1$, where {\em Net} is the graph consisting of six vertices $a, b, c, a', b'$ and $c'$ and six edges $ab, bc, ca, aa', bb'$ and $cc'$ (see Figure~\ref{fig:Gk}).

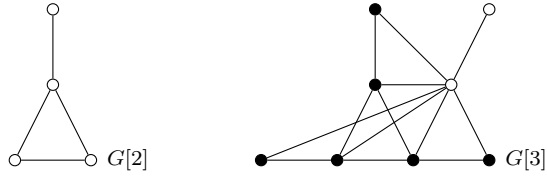
\begin{figure}[htb]
\begin{center}
\begin{tikzpicture}[scale=.5]
\tikzstyle{vertex}=[draw,circle,inner sep=1.5pt]
\node[vertex] (1) at (0,0) {};
\node[vertex] (2) at (2,0) {};
\node[vertex] (3) at (1,2) {};
\node[vertex] (4) at (1,4) {};

\draw (1) -- (2) -- (3) -- (4); \draw (1) -- (3);
\node[] (name) at (2,0) [label=right:{$G[2]$}] {};
\end{tikzpicture}
\qquad\qquad
\begin{tikzpicture}[scale=.5]
\tikzstyle{net}=[draw,circle,inner sep=1.5pt,fill=black];
\tikzstyle{vertex}=[draw,circle,inner sep=1.5pt]
\node[net] (1) at (0,0) {};
\node[net] (2) at (2,0) {};
\node[net] (3) at (4,0) {};
\node[net] (4) at (6,0) {};
\node[net] (5) at (3,2) {};
\node[net] (6) at (3,4) {};
\node[vertex] (7) at (5,2) {};
\node[vertex] (8) at (6,4) {};

\draw (1) -- (2) -- (3) -- (4);
\draw (2) -- (5) -- (6); \draw (3) -- (5);
\draw (1) -- (7) -- (2);
\draw (3) -- (7) -- (4);
\draw (5) -- (7) -- (6);
\draw (8) -- (7);

\node[] (name) at (6,0) [label=right:{$G[3]$}] {};
\end{tikzpicture}
\end{center}
\caption{The graph $G[2]=(K_2+K_1)\star K_1$ and $G[3]=(\text{Net}+K_1)\star K_1$. The black vertices in $G[3]$ induce the Net.}
\label{fig:Gk}
\end{figure}
\begin{proposition}\label{prop:Gk}
$cow(G[k])=k$.
\end{proposition}
\begin{proof}
The claim is obvious in case $k=1$. So, let $k\ge 2$. Note first that $\{\{i\}\mid 1\le i\le k\}$ is a clique in $G[k]$ and the vertex $[k]$ is non-adjacent to all vertices in this clique. Thus, for each $1\le i\le k$, any complete witness for $G[k]$ must contain an independent set containing the two vertices $[k]$ and $\{i\}$. Therefore, any complete witness for $G[k]$ must have at least $k$ independent sets, hence $cow(G[k])\ge k$. On the other hand, the $k$ independent sets $\N_i:=\{M\subseteq [k]\mid i\in M\}$, $1\le i\le k$, form a complete witness for $G[k]$: if $M\not=L\subseteq [k]$ are two non-adjacent vertices in $G[k]$, {\em i.e.}, $M\cap L\not= \emptyset$, then $M,L\in \N_i$ for any $i\in M\cap L$. Hence $cow(G[k])\le k$.
\qed
\end{proof}

Note that in case of $k\ge 2$, the empty set is the unique universal vertex of $G[k]$.
For technical reason, we say that $G[1]=K_2$ has only one universal vertex.
{\em Substituting} a vertex $v$ in a graph $G$ by a graph $H$ results in the graph obtained from $(G-v)\cup H$ by adding all edges between vertices in $N_G(v)$ and vertices in $H$.
We now are able to characterize $k$-probe complete graphs as follows.

\begin{lemma}\label{lem:Gk}
A graph is a $k$-probe complete graph if and only if it is obtained from $G[k]$ by substituting the universal vertex by a (possibly empty) clique and other vertices by (possibly empty) independent sets.
\end{lemma}
\begin{proof}
First, assume that $G$ is a $k$-probe complete graph, and let $Q$ be the set of all universal vertices of $G$ (possibly $Q=\emptyset$). By Proposition~\ref{prop:basic}, $G-Q$ is $k$-probe complete. Let $\N_1,\ldots,\N_k$ be a complete witness for $G-Q$ with $k=cow(G)$. For each $M\subseteq [k]$, $M\not=\emptyset$, let
$$
I_M=\{v\in V(G)\mid v\in\bigcap_{i\in M} \N_i \setminus \bigcup_{j\not\in M} \N_j\}.
$$
Then, as $G-Q$ has no universal vertex, $V(G)\setminus Q= \bigcup_{M}I_M$ is a partition in pairwise disjoint (possible empty) independent sets $I_M$. Observe that, for any non-empty $M,L\subseteq [k]$,
$$
M\cap L\not=\emptyset\, \Leftrightarrow\, \text{no vertex in $I_M$ is adjacent to a vertex in $I_L$ and vice versa}.
$$
Moreover, as $\N_1,\ldots,\N_k$ form a complete witness for $G$, we have
\begin{align*}
M\cap L = & \emptyset\, \Leftrightarrow\, \\
          & \text{every vertex in $I_M$ is adjacent to every vertex in $I_L$ and vice versa}.
\end{align*}
Now, set $I_\emptyset:= Q$ and let $G^*$ be obtained from $G$ by shrinking each $I_M$ to a single vertex $v_M$, $M\subseteq [k]$. Then the facts above show that $G^*$ is isomorphic to $G[k]$ via the bijection $v_M\mapsto M$, and thus, $G$ is obtained from $G[k]$ by substituting the universal vertex $\emptyset$ by the clique $Q$ and other vertices $M$ by the independent sets $I_M$.

For the other direction, suppose that $G$ is obtained from $G[k]$ by substituting the universal vertex by a (possibly empty) clique and other vertices $v$ by (possibly empty) independent sets $I_v$. Then $G$ is a $k$-probe complete graph. Indeed, by Proposition~\ref{prop:Gk}, $G[k]$ is a $k$-probe complete graph. Let $\N_1,\ldots, \N_k$ be a complete witness for $G[k]$. Then the independent sets $\N_i':= \bigcup_{v\in \N_i} I_v$, $1\le i\le k$, form a complete witness for $G$. Consider two arbitrary non-adjacent vertices $x\not=y$ of $G$. If $x,y\in I_v$ for some $v\in V(G[k])$, then $v$ is not the universal vertex of $G[k]$, hence $v\in \N_i$ for some $1\le i\le k$ and therefore $x,y\in \N_i'$. If $x\in I_u$ and $y\in I_v$ for some $u\not=v\in V(G[k])$, then, as $x$ and $y$ are non-adjacent in $G$, $u$ and $v$ are non-adjacent in $G[k]$. Hence $u,v\in\N_i$ for some $1\le i\le k$ and therefore $x,y\in\N_i'$.
\qed
\end{proof}


\begin{theorem}\label{thm:kernel}
\textsc{complete width} (and hence \textsc{edge clique cover}) admits a problem kernel of at most $2^k$ vertices.
\end{theorem}
\begin{proof}
Let $(G,k)$ be an instance of \textsc{complete width}.
By Proposition~\ref{prop:basic}, we may assume that $G$ has no universal vertices and $N(u)\not= N(v)$ for any non-adjacent vertices $u, v$.
Thus, by Lemma~\ref{lem:Gk}, $G$ is (isomorphic to) an induced subgraph of $G[k]$, whenever $G$ is a $k$-probe complete graph.
Since $G[k]$ has $2^k$ vertices, Theorem~\ref{thm:kernel} follows.
\qed
\end{proof}

We remark that it was shown in \cite{CKPPW} that \textsc{edge clique cover}, hence \textsc{complete width}, has no kernel of polynomial size, unless certain complexity assumption fails.

\section{Graphs of small complete width}
We describe in this section graphs of small complete width $k\le 3$. These are particularly $2K_2$-free and our  descriptions are good in the sense that they imply polynomial-time recognition for these graph classes.

\subsection{Complete width-1 and complete width-2 graphs}

A \emph{complete split graph} is a split graph $G=(Q+S,E)$ such that every vertex in the clique $Q$ is adjacent to every vertex in the independent set $S$. Such a partition is also called a \emph{complete split partition} of a split graph. Note that if the complete split graph $G=(Q+S,E)$ is not a clique, then $G$ has exactly one complete split partition $V=Q\cup S$. Furthermore, each vertex in $Q$, if any, is a universal vertex.

Graphs of complete width one can be characterized as follows.

\begin{theorem}\label{thm:completesplit}
The following statements are equivalent. 
 \begin{itemize}
  \item[\em (i)] $G$ is a probe complete graph;
  \item[\em (ii)] $G$ is a $(K_2+K_1, C_4)$-free graph;
  \item[\em (iii)] $G$ is a complete split graph;
  \item[\em (iv)] $G$ is obtained from a $K_2$ by substituting one vertex by a clique and the other vertex by an independent set.
 \end{itemize}
\end{theorem}
\begin{proof}
The equivalence of (i), (ii) and (iii) has been shown in~\cite{LePeng12}. The equivalence of (i) and (iv) follows from Lemma~\ref{lem:Gk}.
\qed
\end{proof}

Graphs of complete width at most~two can be characterized as follows.

\begin{theorem}\label{thm:2-probecomplete}
The following statements are equivalent. 
 \begin{itemize}
  \item[\em (i)]  $G$ is a $2$-probe complete graph;
  \item[\em (ii)] $G$ is $(2K_2, P_4, K_3+ K_1, (K_2+K_1)\star 2K_1, C_4\star 2K_1)$-free; see Fig.~\ref{fig:2-probecomplete};
  \item[\em (iii)] $G$ is obtained from $G[2]=(K_2+K_1)\star K_1$ by substituting the universal vertex by a clique and the other vertices by independent sets.
 \end{itemize}
\end{theorem}
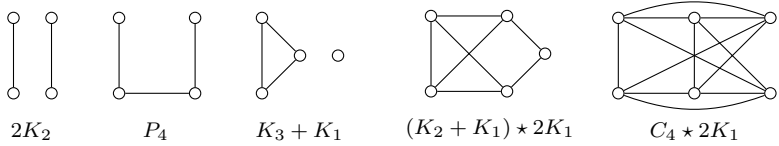
\begin{figure}[htb]
\begin{center}
\begin{tikzpicture}[scale=.5]
\tikzstyle{vertex}=[draw,circle,inner sep=1.5pt]
\node[vertex] (1) at (0,0) {};
\node[vertex] (2) at (0,2) {};
\node[vertex] (3) at (1,2) {};
\node[vertex] (4) at (1,0) {};

\draw (1) -- (2);
\draw (3) -- (4);
\node[] (name) at (-0.5,-1) [label=right:{$2K_2$}] {};
\end{tikzpicture}
\qquad
\begin{tikzpicture}[scale=.5]
\tikzstyle{vertex}=[draw,circle,inner sep=1.5pt]
\node[vertex] (1) at (0,0) {};
\node[vertex] (2) at (0,2) {};
\node[vertex] (3) at (2,2) {};
\node[vertex] (4) at (2,0) {};

\draw (2) -- (1) -- (4) -- (3);
\node[] (name) at (0.2,-1) [label=right:{$P_4$}] {};
\end{tikzpicture}
\quad
\begin{tikzpicture}[scale=.5]
\tikzstyle{vertex}=[draw,circle,inner sep=1.5pt]
\node[vertex] (1) at (0,0) {};
\node[vertex] (2) at (0,2) {};
\node[vertex] (3) at (1,1) {};
\node[vertex] (4) at (2,1) {};

\draw (1) -- (2) -- (3) -- (1);
\node[] (name) at (-0.6,-1) [label=right:{$K_3+K_1$}] {};
\end{tikzpicture}
\quad
\begin{tikzpicture}[scale=.5]
\tikzstyle{vertex}=[draw,circle,inner sep=1.5pt]
\node[vertex] (1) at (1,0) {};
\node[vertex] (3) at (1,2) {};
\node[vertex] (4) at (3,2) {};
\node[vertex] (5) at (4,1) {};
\node[vertex] (6) at (3,0) {};

\draw (3) -- (4) -- (5) -- (6) -- (1);
\draw (4) -- (1) -- (3) -- (6);
\node[] (name) at (-0.1,-1) [label=right:{$(K_2+K_1)\star 2K_1$}] {};
\end{tikzpicture}
\quad
\begin{tikzpicture}[scale=.5]
\tikzstyle{vertex}=[draw,circle,inner sep=1.5pt]
\node[vertex] (1) at (0,0) {};
\node[vertex] (2) at (0,2) {};
\node[vertex] (3) at (2,2) {};
\node[vertex] (4) at (2,0) {};
\node[vertex] (5) at (4,0) {};
\node[vertex] (6) at (4,2) {};

\draw (1) -- (2) -- (3) -- (4) -- (1);
\draw (4) -- (5) -- (3) -- (6) -- (4);
\draw (6) -- (1) to[bend angle=20, bend right] (5);
\draw (5) -- (2) to[bend angle=20, bend left] (6);
\node[] (name) at (0.4,-1) [label=right:{$C_4\star 2K_1$}] {};
\end{tikzpicture}
\end{center}
\caption{Forbidden induced subgraphs for $2$-probe complete graphs.}
\label{fig:2-probecomplete}
\end{figure}
\begin{proof}
The equivalence of (i) and (iii) has been shown in~\cite{LePeng14}. The equivalence of (i) and (iii) follows from Lemma~\ref{lem:Gk}.
\qed
\end{proof}

\subsection{Complete width-3 graphs}

Graphs of complete width at most~3 can be characterized as follows.

\begin{theorem}\label{thm:3-probecomplete}
The following statements are equivalent. 
 \begin{itemize}
  \item[\em (i)]  $G$ is a $3$-probe complete graph;
  \item[\em (ii)] $G$ is $(F_1,\ldots, F_{14})$-free; see Fig.~\ref{fig:3-probecomplete};
  \item[\em (iii)] $G$ is obtained from $G[3]=(\text{Net}+K_1)\star K_1$ by substituting the universal vertex by a clique and the other vertices by independent sets.
 \end{itemize}
\end{theorem}
\begin{figure}[htb]
\begin{center}
\begin{tikzpicture}[scale=.5]
\tikzstyle{vertex}=[draw,circle,inner sep=1.5pt]
\node[vertex] (1) at (0,0) {};
\node[vertex] (2) at (0,2) {};
\node[vertex] (3) at (1,2) {};
\node[vertex] (4) at (1,0) {};

\draw (1) -- (2);
\draw (3) -- (4);
\node[] (name) at (-1.1,-1) [label=right:{$F_1 = 2K_2$}] {};
\end{tikzpicture}
\qquad
\begin{tikzpicture}[scale=.5]
\tikzstyle{vertex}=[draw,circle,inner sep=1.5pt]
\node[vertex] (1) at (.8,0) {};
\node[vertex] (2) at (0.2,1.2) {};
\node[vertex] (3) at (1.5,2) {};
\node[vertex] (4) at (2.8,1.2) {};
\node[vertex] (5) at (2.2,0) {};

\draw (1) -- (2) -- (3) -- (4) -- (5) -- (1);
\node[] (name) at (0,-1) [label=right:{$F_2 = C_5$}] {};
\end{tikzpicture}
\qquad
\begin{tikzpicture}[scale=.5]
\tikzstyle{vertex}=[draw,circle,inner sep=1.5pt]
\node[vertex] (1) at (.8,0) {};
\node[vertex] (2) at (0.2,1.2) {};
\node[vertex] (3) at (1.5,2) {};
\node[vertex] (4) at (2.8,1.2) {};
\node[vertex] (5) at (2.2,0) {};

\draw (1) -- (2) -- (3) -- (4) -- (5) -- (1);
\draw (2) -- (4);
\node[] (name) at (0,-1) [label=right:{$F_3 = \overline{P_5}$}] {};
\end{tikzpicture}
\qquad
\begin{tikzpicture}[scale=.5]
\tikzstyle{vertex}=[draw,circle,inner sep=1.5pt]
\node[vertex] (1) at (0,0) {};
\node[vertex] (2) at (0,2) {};
\node[vertex] (3) at (2,2) {};
\node[vertex] (4) at (2,0) {};
\node[vertex] (5) at (3,0) {};

\draw (1) -- (2);
\draw (3) -- (4);
\draw (1) -- (3) -- (2) -- (4) -- (1);
\node[] (name) at (-1,-1) [label=right:{$F_4 = K_4+K_1$}] {};
\end{tikzpicture}
\end{center}

\begin{center}
\begin{tikzpicture}[scale=.5]
\tikzstyle{vertex}=[draw,circle,inner sep=1.5pt]
\node[vertex] (1) at (0,0) {};
\node[vertex] (2) at (0,2) {};
\node[vertex] (3) at (2,2) {};
\node[vertex] (4) at (2,0) {};
\node[vertex] (5) at (4,0) {};

\draw (3) -- (1) -- (2) -- (3) -- (4) -- (1);
\draw (4) -- (5);
\node[] (name) at (1.2,-1) [label=right:{$F_5$}] {};
\end{tikzpicture}
\qquad
\begin{tikzpicture}[scale=.5]
\tikzstyle{vertex}=[draw,circle,inner sep=1.5pt]
\node[vertex] (1) at (0,0) {};
\node[vertex] (2) at (0,2) {};
\node[vertex] (3) at (2,2) {};
\node[vertex] (4) at (2,0) {};
\node[vertex] (5) at (4,0) {};
\node[vertex] (6) at (4,2) {};

\draw (1) -- (2) -- (3) -- (4) -- (1);
\draw (4) -- (5);
\draw (3) -- (6);
\node[] (name) at (1.2,-1) [label=right:{$F_6$}] {};
\end{tikzpicture}
\qquad
\begin{tikzpicture}[scale=.5]
\tikzstyle{vertex}=[draw,circle,inner sep=1.5pt]
\node[vertex] (1) at (0,0) {};
\node[vertex] (2) at (0,2) {};
\node[vertex] (3) at (2,2) {};
\node[vertex] (4) at (2,0) {};
\node[vertex] (5) at (4,0) {};
\node[vertex] (6) at (4,2) {};

\draw (1) -- (2) -- (3) -- (4) -- (1); \draw (1) -- (3); \draw (2) -- (4);
\draw (4) -- (5);
\draw (3) -- (6);
\node[] (name) at (1.2,-1) [label=right:{$F_7$}] {};
\end{tikzpicture}
\end{center}

\begin{center}
\begin{tikzpicture}[scale=.5]
\tikzstyle{vertex}=[draw,circle,inner sep=1.5pt]
\node[vertex] (1) at (0,0) {};
\node[vertex] (2) at (1,0) {};
\node[vertex] (3) at (2,0) {};
\node[vertex] (4) at (3,0) {};
\node[vertex] (5) at (1.5,2) {};
\node[vertex] (6) at (4,0) {};

\draw (5) -- (1) -- (2) -- (3) -- (4) -- (5);
\draw (2) -- (5) -- (3);
\node[] (name) at (-1.2,-1) [label=right:{$F_8=(P_4\star K_1)+K_1$}] {};
\end{tikzpicture}
\quad
\begin{tikzpicture}[scale=.5]
\tikzstyle{vertex}=[draw,circle,inner sep=1.5pt]
\node[vertex] (1) at (0,0.5) {};
\node[vertex] (2) at (1,0.5) {};
\node[vertex] (3) at (2,0.5) {};
\node[vertex] (4) at (3,0.5) {};
\node[vertex] (5) at (1.5,2) {};
\node[vertex] (6) at (1.5,-1) {};

\draw (5) -- (1) -- (2) -- (3) -- (4) -- (5);
\draw (2) -- (5) -- (3);
\draw (2) -- (6) -- (3);
\node[] (name) at (1.5,-1) [label=right:{$F_9$}] {};
\end{tikzpicture}
\quad
\begin{tikzpicture}[scale=.5]
\tikzstyle{vertex}=[draw,circle,inner sep=1.5pt]
\node[vertex] (1) at (0,0.5) {};
\node[vertex] (2) at (1,0.5) {};
\node[vertex] (3) at (2,0.5) {};
\node[vertex] (4) at (3,0.5) {};
\node[vertex] (5) at (1.5,2) {};
\node[vertex] (6) at (1.5,-1) {};

\draw (5) -- (1) -- (2) -- (3) -- (4) -- (5);
\draw (2) -- (5) -- (3);
\draw (1) -- (6) -- (4);
\draw (2) -- (6) -- (3);
\node[] (name) at (1.6,-1) [label=right:{$F_{10}$}] {}; 
\end{tikzpicture}
\quad
\begin{tikzpicture}[scale=.5]
\tikzstyle{vertex}=[draw,circle,inner sep=1.5pt]
\node[vertex] (1) at (1,0) {};
\node[vertex] (2) at (0,1) {};
\node[vertex] (3) at (1,2) {};
\node[vertex] (4) at (3,2) {};
\node[vertex] (5) at (4,1) {};
\node[vertex] (6) at (3,0) {};

\draw (1) -- (2) -- (3) -- (4) -- (5) -- (6) -- (1);
\draw (4) -- (2) -- (6);
\draw (4) -- (1) -- (3) -- (6);
\node[] (name) at (-1.2,-1) [label=right:{$F_{11}=(K_3+K_1)\star 2K_1$}] {};
\end{tikzpicture}
\end{center}

\begin{center}
\begin{tikzpicture}[scale=.5]
\tikzstyle{vertex}=[draw,circle,inner sep=1.5pt]
\node[vertex] (1) at (0,0) {};
\node[vertex] (2) at (0,2) {};
\node[vertex] (3) at (2,2) {};
\node[vertex] (4) at (2,0) {};
\node[vertex] (5) at (4,0) {};
\node[vertex] (6) at (4,2) {};

\draw (1) -- (2) -- (3) -- (4) -- (1); \draw (1) -- (3); \draw (2) -- (4);
\draw (4) -- (5);
\draw (3) -- (6);
\draw (5) -- (6);
\draw (1) to[bend angle=20, bend right] (5);
\draw (2) to[bend angle=20, bend left] (6);
\node[] (name) at (-0.3,-1) [label=right:{$F_{12}=\overline{P_3}\star\overline{P_3}$}] {};
\end{tikzpicture}
\qquad
\begin{tikzpicture}[scale=.5]
\tikzstyle{vertex}=[draw,circle,inner sep=1.5pt]
\node[vertex] (1) at (0,0) {};
\node[vertex] (2) at (0,2) {};
\node[vertex] (3) at (2,2) {};
\node[vertex] (4) at (2,0) {};
\node[vertex] (5) at (4,0) {};
\node[vertex] (6) at (4,2) {};
\node[vertex] (7) at (5,1) {};

\draw (1) -- (2) -- (3) -- (4) -- (1); \draw (1) -- (3); \draw (2) -- (4);
\draw (5) -- (6);
\draw (3) -- (7) -- (4);
\draw (5) -- (7) -- (6);
\draw (2) -- (5) -- (4);
\draw (1) -- (6) -- (3);
\draw (1) to[bend angle=20, bend right] (5);
\draw (2) to[bend angle=20, bend left] (6);
\node[] (name) at (-.8,-1) [label=right:{$F_{13}=(K_2+K_1)\star C_4$}] {};
\end{tikzpicture}
\quad
\begin{tikzpicture}[scale=.5]
\tikzstyle{vertex}=[draw,circle,inner sep=1.5pt]
\node[vertex] (1) at (0,1) {};
\node[vertex] (2) at (0,2) {};
\node[vertex] (3) at (1,3) {};
\node[vertex] (4) at (2,3) {};
\node[vertex] (5) at (3,2) {};
\node[vertex] (6) at (3,1) {};
\node[vertex] (7) at (2,0) {};
\node[vertex] (8) at (1,0) {};

\draw (1) -- (3); \draw (1) -- (4); \draw (1) -- (5); \draw (1) -- (6); \draw (1) -- (7); \draw (1) -- (8);
\draw (2) -- (3); \draw (2) -- (4); \draw (2) -- (5); \draw (2) -- (6); \draw (2) -- (7); \draw (2) -- (8);

\draw (3) -- (5); \draw (3) -- (6); \draw (3) -- (7); \draw (3) -- (8);
\draw (4) -- (5); \draw (4) -- (6); \draw (4) -- (7); \draw (4) -- (8);

\draw (5) -- (7); \draw (5) -- (8);
\draw (6) -- (7); \draw (6) -- (8);

\node[] (name) at (-.8,-1) [label=right:{$F_{14}=C_4\star C_4$}] {};
\end{tikzpicture}
\end{center}
\caption{Forbidden induced subgraphs for $3$-probe complete graphs.}
\label{fig:3-probecomplete}
\end{figure}
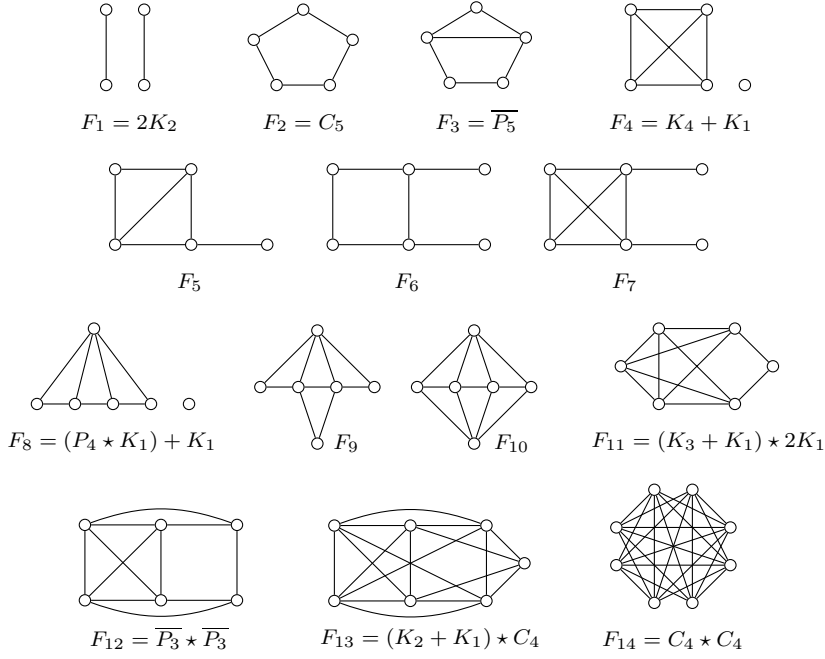
\begin{proof} By Lemma~\ref{lem:Gk} it remains to prove the equivalence of (i) and (ii).

\medskip
\noindent
(i) $\Rightarrow$ (ii):
By inspection one can easily see that none of the graphs depicted in Fig.\ref{fig:3-probecomplete} is a $3$-probe complete graph. Thus, no $3$-probe complete graph contains any of these graphs as an induced subgraph.

\medskip
\noindent
(ii) $\Rightarrow$ (i): Let $G$ be a $(F_1,\ldots, F_{14})$-free graph. Let $Q$ be the set of all universal vertices of $G$. As $G$ is $F_1$-free, $G-Q$ has at most one non-trivial connected component. Let $H$ be the non-trivial connected component of $G-Q$ (if $H$ does not exist, $G$ is a $1$-probe complete graph and we are done), and let $I$ be the set of all isolated vertices of $G-Q$.

We distinguish two cases; note that, as $G$ is $F_1$-free, $G$ is particularly $P_5$-free.

\medskip\noindent
{\em Case 1.\, $H$ contains an induced $P_4$.}\, Let $P= v_1v_2v_3v_4$ be an induced $P_4$ in $H$ with edges $v_1v_2, v_2v_3$ and $v_3v_4$. For each $S\subseteq\{1,2,3,4\}$ write
$$
M_S=\left\{v\mid v\in V(H)\setminus V(P), N(v)\cap V(P)=\{v_i\mid i\in S\}\right\},
$$
that is, $M_S$ consists of all vertices of $H$ outside $P$ adjacent exactly to $v_i, i\in S$. By definition, $M_S\cap M_{S'}=\emptyset$ whenever $S\not=S'$. To simplify the notion, we also write $M_0$ for $M_\emptyset$, $M_3$ for $M_{\{3\}}$ and $M_{124}$ for $M_{\{1,2,4\}}$ and so on. We have the following facts.
\begin{itemize}
\item $M_1=M_4=M_{12}=M_{34}=M_{14}=\emptyset$. This is because $G$ is $(F_1,F_2)$-free.
\item $M_{124}=M_{134}=M_{123}=M_{234}=\emptyset$. This is because $G$ is $(F_3,F_5)$-free.
\item $M_{1234}=\emptyset$. Assume that $M_{1234}\not=\emptyset$. Then $M_{1234}$ is a clique (as $G$ is $F_{10}$-free), and every vertex in $M_{1234}$ is adjacent to all vertices in $M_2\cup M_3\cup M_{13}\cup M_{24}\cup M_{23}\cup M_0\cup I$ (as $G$ is $(F_5,F_3,F_9,F_8)$-free). But then $M_{1234}\subseteq Q$, a contradiction.
\end{itemize}
Thus, $V(H)=M_0\cup M_2\cup M_3\cup M_{13}\cup M_{24}\cup M_{23}$. Moreover,
\begin{itemize}
\item $M_0, M_2, M_3, M_{13}, M_{24}$ and $M_{23}$ are independent sets. This is because $G$ is $(F_1,F_5,F_7)$-free.
\end{itemize}
For two disjoint sets $U,W$ of vertices, we write $U\join W$, respectively $U\co-join W$, to describe the fact that every vertex in $U$ is adjacent, respectively non-adjacent, to every vertex in $W$ and vice versa. We have the following facts.
\begin{itemize}
\item $M_0\co-join (M_2\cup M_3)$ and $M_0\co-join (M_{13}\cup M_{24})$. These are because $G$ is $(F_1,F_6)$-free.
\item $M_0\join M_{23}$. Since $M_0$ is independent and $M_0\co-join (M_2\cup M_3\cup M_{13}\cup M_{24})$, the connectedness of $H$ implies that each vertex in $M_0$, if any, must have a neighbor in $M_{23}$. Since $M_{23}$ is independent and $G$ is $F_5$-free, each vertex in $M_0$ therefore is adjacent to all vertices in $M_{23}$.
\item $M_2\co-join M_3$, $M_2\join M_{13}$ (as $G$ is $F_6$-free) and $M_2\co-join (M_{23}\cup M_{24})$ (as $G$ is $(F_5, F_3)$-free).
\item $M_3\join M_{24}$ and $M_3\co-join (M_{23}\cup M_{13})$. These are obtained by symmetry.
\item $M_{13}\join M_{24}$. This is because $G$ is $F_1$-free.
\end{itemize}
Thus, the three independent sets
\begin{eqnarray*}
\N_1:= & M_2\cup\{v_1\}\cup M_3\cup\{v_4\}\cup M_{23}\cup I,\\
\N_2:= & M_0\cup M_2\cup\{v_1\}\cup M_{24}\cup\{v_3\}\cup I,\\
\N_3:= & M_0\cup M_3\cup\{v_4\}\cup M_{13}\cup\{v_2\}\cup I,
\end{eqnarray*}
form a complete witness for $G$, and Case 1 is settled.

\medskip\noindent
{\em Case 2.\, $H$ is $P_4$-free.}\, That is, $H$ is a cograph. It is a well-known fact that any connected cograph is the join of two smaller cographs (see, {\em e.g.}, \cite{BraLeSpi}). This fact immediately implies that any connected $C_4$-free cograph has a universal vertex.

Now, let $H=H_1\star H_2$. Then $H_1$ or $H_2$ is edgeless. To see this, suppose the contrary that both $H_1$ and $H_2$ have some edges. Then $I=\emptyset$ (otherwise $G$ would have an $F_4$), hence neither $H_1$ nor $H_2$ can have a universal vertex (otherwise $G-Q$ would have a universal vertex). Moreover, $H_1$ or $H_2$ must be connected (otherwise both $H_1$ and $H_2$ would have a $\overline{P_3}$, and $G$ would have an $F_{12}$). Let, say, $H_1$ be connected. Then, as $H_1$ has no universal vertex, $H_1$ has a $C_4$. Now, if $H_2$ is disconnected, then $G$ has an $F_{13}$. If $H_1$ is connected, then, as $H_2$ has no universal vertex, $H_2$ has a $C_4$. But then $G$ has an $F_{14}$. This contradiction shows that $H_1$ or $H_2$ must be edgeless, as claimed. Say, without loss of generality,
\begin{equation*}
\text{$H_1$ is edgeless.}
\end{equation*}
We distinguish two cases.

\medskip\noindent
{\em Case 2.1.\, $I=\emptyset$.}\, Then the independent set $V(H_1)$ has at least two vertices (otherwise the vertex of $H_1$ would be a universal vertex of $G$). Hence $H_2$ is
\begin{itemize}
\item $(K_3+K_1)$-free (otherwise $G$ would have an $F_{11}$),
\item $(K_2+K_1)\star 2K_1$-free (otherwise $G$ would have an $F_{13}$), and
\item $(C_4\star 2K_1$)-free (otherwise $G$ would have an $F_{14}$).
\end{itemize}
Thus, by Theorem~\ref{thm:2-probecomplete}, $H_2$ is a $2$-probe complete graph. Let $\N_1, \N_2$ be a complete witness for $H_2$. Then $\N_1, \N_2$ and $\N_3:=V(H_1)$ clearly from a complete witness for $G$.

\medskip\noindent
{\em Case 2.2.\, $I\not=\emptyset$.}\, Then $H_2$ is $K_3$-free (otherwise a $K_3$ in $H_2$, a vertex in $H_1$ and a vertex in $I$ would induce an $F_4$). By Theorem~\ref{thm:2-probecomplete}, $H_2$ is a $2$-probe complete.

Suppose first that $H_2$ has a universal vertex $v$. Then $V(H_2)\setminus\{v\}$ is an independent set, and $\N_1:=V(H_1)\cup I$, $\N_2:= (V(H_2)\setminus\{v\})\cup I$ and $\N_3:=I\cup\{v\}$ clearly form a complete witness for $G$.

Suppose next that $H_2$ has no universal vertex. Recall that $H_2$ is a $2$-probe complete graph, and consider a complete witness $\N_1, \N_2$ for $H_2$. Since $H_2$ has no universal vertex, any vertex of $H_2$ must belong to $\N_1$ or $\N_2$. Thus, $\N_1':= \N_1\cup I$, $\N_2':= \N_2\cup I$ and $\N_3:= V(H_1)\cup I$ clearly form a complete witness for $G$.

\medskip
Case 2 is settled, and the proof of Theorem~\ref{thm:3-probecomplete} is complete.
\qed
\end{proof}

We note that, by using modular decomposition (see, {\em e.g.}, \cite{HabPau,TCHP}), one can recognize graphs obtained from the Net by substituting vertices by independent sets in linear time. Hence Theorem~\ref{thm:3-probecomplete} gives a linear time recognition for 3-probe complete graphs.

\section{Conclusion}\label{sec:conclusion}
In this paper we have shown that \textsc{complete width} is NP-complete on $3K_2$-free bipartite graphs (equivalently, \textsc{edge clique cover} is NP-complete on $\overline{3K_{2}}$-free co-bipartite graphs). So, an obvious open question is: What is the computational complexity of \textsc{complete width} on $2K_2$-free graphs? Equivalently, what is the computational complexity of \textsc{edge clique cover} on $C_4$-free graphs? We have given partial results in this direction by showing that \textsc{complete width} is polynomially solvable on $(2K_2,K_3)$-free graphs and on $(2K_2,C_4)$-free graphs. (Equivalently, \textsc{edge clique cover} is polynomially solvable on $(C_4, 3K_1)$-free graphs and on $(C_4, 2K_2)$-free graphs.)

Another interesting question is the following. The time complexities of many problems coincide on split graphs and bipartite graphs, \emph{e.g.}, the dominating set problem. However, for the complete width problem, they are different, one is in P and the other is in NP-complete. Trees are a special class of bipartite graphs. Many problems become easy on trees. However, we do not know the hardness of the complete width problem on trees.



\begin{thebibliography}{}
%
%
\bibitem{BHPT}
Bl\'azsik, Z., Hujter, M., Pluh\'ar, A., Tuza, Z.:
Graphs with no induced $C_4$ and $2K_2$.
{\em Discrete Mathematics} 115, 51--55 (1993)

\bibitem{BraLeSpi}
Brandst\"adt, A., Le, V.B., Spinrad, J.P.:
 \emph{Graph Classes: A Survey}.
 SIAM Monographs on Discrete Mathematics and Applications, Philadelphia (1999)

\newcommand{\mytilde}{\raise.17ex\hbox{$\scriptstyle\mathtt{\sim}$}}

\bibitem{ChaChaKloPen}
Chandler, D.B., Chang, M.-S., Kloks, T., Peng, S.-L.:
Probe Graphs. Manuscript,
http:/\!/www.cs.ccu.edu.tw/\mytilde hunglc/ProbeGraphs.pdf (2009)

\bibitem{ChaHunKloPen}
Chang, M.-S., Hung, L.-J., Kloks, T., Peng, S.-L.:
Block-graph width.
 {\em Theoretical Computer Science} 412, 2496--2502 (2011)

\bibitem{CHL}
Chang, M.-S., Kloks, T., Liu, C.-H.:
Edge-clique graphs of cocktail parties have unbounded rankwidth.
\textrm{arXiv:1205.2483 [cs.DM]} (2012)

\bibitem{CM}
Chang, M.-S., M\"uller, H.:
On the tree-degree of graphs.
In: Proc. 27th WG. LNCS 2204, 44--54 (2008)

\bibitem{CFKLMPPS}
 Cygan, M., Fomin, F.V., Kowalik, L., Lokshtanov, D., Marx, D., Pilipczuk, M., Pilipczuk, M., Saurabh, S.:
{\em Parameterized Algorithms}.
Springer (2015)

\bibitem{CKPPW}
Cygan, M., Kratsch, S., Pilipczuk, M., Pilipczuk, M., Wahlstr\"om, M.:
Clique cover and graph separation: New incompressibility results.
{\em ACM Transaction on Computation Theory} 6 (2014)

\bibitem{CPP}
Cygan, M., Pilipczuky, M., Pilipczuk, M.:
Known algorithms for EDGE CLIQUE COVER are probably optimal.
In: Proc. SODA, 1044--1053 (2013)

\bibitem{FolHam}
Foldes, S., Hammer, P.L.:
 Split graphs.
{\em Congressus Numerantium}, No. XIX, 311--315 (1977)

\bibitem{Golumbic}
Golumbic, M.C.:
\emph{Algorithmic Graph Theory and Perfect Graphs} (Second edition).
Annals of Discrete Math. 57, Elsevier, Amsterdam (2004)

\bibitem{GolTre}
Golumbic, M.C, Trenk, A.N.: {\em Tolerance Graphs\/}. Cambridge studies in advanced mathematics 89, New York (2004)

\bibitem{GGHN}
Gramm, J., Guo, J., H\"uffner, F., Niedermeier, R.:
Data reduction and exact algorithms for clique cover.
{\em ACM Journal of Experimental Algorithmics} 13, Article 2.2 (2008)

\bibitem{HabPau}
Habib, M., Paul, C.:
A survey of the algorithmic aspects of modular
decomposition.
{\em Computer Science Review} 4, 41--59 (2010)

\bibitem{HamPelSun}
 Hammer, P.L., Peled, U.N., Sun, X.:
 Difference graphs.
 {\em Discrete Applied Mathematics} 28, 35--44 (1990)

\bibitem{Holyer}
Holyer, I.:
The NP-completeness of some edge-partition problems.
\emph{SIAM Journal on Computing} 4, 713--717 (1981)

\bibitem{Hoover}
Hoover, D.N.:
Complexity of graph covering problems for graphs of low degree.
{\em Journal of Combinatorial Mathematics and Combinatorial Computing} 11, 187--208 (1992)

\bibitem{HsuTsa}
Hsu, W.-L., Tsai, K.-H.:
Linear time algorithms on circular-arc graphs.
{\em Inf. Process. Lett.} 40, 123--129 (1991)

\bibitem{KSW}
Kou, L.T., Stockmeyer, L.J., Wong, C.K.:
Covering edges by cliques with regard
to keyword conflicts and intersection graphs.
{\em Comm. ACM} 21, 135--139 (1978)

\bibitem{LePeng12}
Le, V.B., Peng, S.-L.:
 Characterizing and recognizing probe block graphs.
 {\em Theoretical Computer Science} 568, 97--102 (2015)

\bibitem{LePeng14}
Le, V.B., Peng, S.-L.:
Good characterizations and linear time recognition for 2-probe block graphs.
In: {\em Proceedings of the International Computer Symposium}, Taichung, Taiwan, December 12-14, 2014. IOS Press, 22--31 (2015)
doi:10.3233/978-1-61499-484-8-22

\bibitem{cocoon2015}
Le, V.B., Peng, S.-L.:
On the complete width and edge clique cover problems.
In: {\em Proceedings of the 21st International Conference COCOON 2015}.
Lecture Notes in Computer Science 9198, 537--547 (2015)

\bibitem{MWW}
Ma, S., Wallis, W.D., Wu, J.:
 Clique covering of chordal graphs.
 \emph{Utilitas Mathematica} 36, 151--152 (1989)

\bibitem{MafPre}
Maffray, F., Preissmann, M.:
Linear recognition of pseudo-split graphs.
{\em Discrete Applied Mathematics} 52, 307--312 (1994)

\bibitem{MahPel}
Mahadev, N.V.R., Peled, U.N.:
{\em Threshold Graphs and Related Topics}.
Annals of discrete mathematics 56, Elsevier, Amsterdam (1995)

\bibitem{Muller}
M\"uller, H.:
On edge perfectness and classes of bipartite graphs.
\emph{Discrete Math.} 149, 159--187 (1996)

\bibitem{Orlin}
Orlin, J.:
Contentment in graph theory: covering graphs with cliques.
\emph{Indagationes Mathematicae} 80, 406--424 (1977)


\bibitem{Pullman}
Pullman, N.J.:
Clique covering of graphs IV. Algorithms.
\emph{SIAM Journal on Computing} 13, 57--75 (1984)

\bibitem{Ray}
Raychaudhuri, A.:
Intersection number and edge clique graphs of chordal and strongly chordal graphs.
\emph{Congressus Numer.} 67, 197--204 (1988)

\bibitem{TCHP}
Tedder, M., Corneil, D., Habib, M., Paul, C.:
Simpler linear-time modular decomposition
via recursive factorizing permutations. In: Automata, languages and
programming. Lecture Notes in Comput. Sci. 5125, 634--645 (2008)

\bibitem{Yannakakis}
 Yannakakis, M.:
 Node-delection problems on bipartite graphs.
 {\em SIAM Journal on Computing} 10, 310--327 (1981)
\end{thebibliography}


\end{document}